\documentclass[reprint,amsmath,amssymb,aps,pra]{revtex4-2}

\usepackage{dcolumn}
\usepackage{bm}

\usepackage{relsize}
\usepackage[utf8]{inputenc}
\usepackage[english]{babel}
\usepackage[T1]{fontenc}
\usepackage{hyperref}
\usepackage{graphicx}

\usepackage{amsthm}
\usepackage{amsmath} 
\usepackage{amssymb}
 
\newtheorem{theorem}{Theorem}
\newtheorem{corollary}[theorem]{Corollary}
\newtheorem{lemma}[theorem]{Lemma}
\newtheorem{definition}[theorem]{Definition}
\newtheorem{remark}[theorem]{Remark}
\newtheorem{example}[theorem]{Example}

\begin{document}

\preprint{APS/123-QED}

\title{Local Dimension Invariant Qudit Stabilizer Codes}

\author{Lane G. Gunderman}
 \email{lgunderman@uwaterloo.ca}
\affiliation{%
 The Institute for Quantum Computing
}%
\affiliation{University of Waterloo Department of Physics and Astronomy}

\date{\today}

\begin{abstract}
Protection of quantum information from noise is a massive challenge. One avenue people have begun to explore is reducing the number of particles needing to be protected from noise and instead use systems with more states, so called qudit quantum computers. These systems will require codes which utilize the full computational space. Many prior qudit codes are very restrictive on relations between parameters. In this paper we show that codes for these systems can be derived from already known codes, often relaxing the constraints somewhat, a result which could prove to be very useful for fault-tolerant qudit quantum computers.
\end{abstract}
\maketitle
\section{Background}

The ability to perform classical computation within an arbitrarily small error rate was shown by Shannon in the 40's \cite{shannon}. He provided a theoretical framework showing that modern classical computation would be possible. From that point, there arose a new challenge of finding actual codes that could best implement Shannon's result. This in turn pushed coding theory into a new realm, inspiring codes such as the Hamming code family and BCH codes \cite{bch}, and later leading to incredible ideas such as Polar codes \cite{polarog} and Turbo codes \cite{turboog}.

As computational power progressed, there began to be investigations into the potential power of using quantum phenomena as a computational tool. This brought those same questions explored for classical computers back into question. This led to various ideas to try to bring over classical codes in some form or another. Among some of the earlier ideas was the stabilizer formalism \cite{thesis}, CSS codes \cite{css1}\cite{css2}, and teleportation \cite{tele}. Many classical coding theory methods have been generalized into this new quantum setting, such as Polynomial codes (a generalization of BCH and cyclic codes) \cite{poly1}\cite{poly2}, Polar codes \cite{polar}, and Turbo codes \cite{turbo1}\cite{turbo2}--including results such as a complete list of all perfect codes \cite{perfect}.

Building a quantum computer out of qudits (quantum objects with more than two levels) instead of qubits (quantum objects with only two levels) is an appealing option since such a system would need comparatively few qudits to perform large quantum computations, due to the larger computational space of each particle in such a system. In addition, context being the cause for the magic in quantum computation in the qudit case has been shown, whereas the case for qubits is still an open problem \cite{contextuality}. This has led to the characterization of magic-state distillation regions for qudits as well as fault-tolerant methods for such \cite{qutritmsd}\cite{qutritmsdtight}\cite{quditmsd}.

This means that we also need error-correction methods for these qudit systems. Prior work on qudit codes often depends on having a classical code which satisfies the conditions needed for CSS code construction, or a similar orthogonality requirement (such as \cite{quditgeneral}\cite{quditbch}\cite{quditmds}). This allows for the generation of many qudit quantum codes, however, at times these codes can require very strict relations between the number of bases for the particles, the number of particles, and the number of logical qudits. This can result in these codes being less useful for constructed qudit systems. This work aims to tackle this problem by working to reduce this level of restriction by allowing codes to be used for qudits of a different number of bases than they were initially designed for. In some regards one may consider this a tool somewhat similar in nature to CSS code construction: CSS allows classical to quantum code construction whereas this allows for quantum to quantum code construction. In addition, this work may provide an avenue for determining whether a code is utilizing the qudit space particularly well.

Experimental realizations of qudit quantum computers have been progressing as well as the theory of making such systems \cite{qudition}\cite{quditlight}. As these systems come online and grow there will be a need to have more flexibility in the set of codes that can be used to protect the information in these systems. In this article we primarily explore the ability to apply quantum error-correcting codes in smaller dimensional spaces onto systems with larger alphabets without having to discover codes for those systems through other methods, thus creating extensions of these already known codes into larger spaces.

Before we move on to discussing this problem, we must first define our mathematical language for working on these problems. Following that we introduce our results showing the ability to apply codes in larger spaces then show the condition required for preserving the distance of such codes as well as a region where the distance of these codes can be preserved. We then propose some directions to carry this work.

\section{Definitions}
In this section we define the majority of the tools used in this paper. We recall common definitions and results for qudit operators.

A qubit is defined as a two level system with states $|0\rangle$ and $|1\rangle$. We define a qudit as being a quantum system over $q$ levels, where $q$ is prime.

\begin{definition}
Generalized Paulis for a space over $q$ orthogonal levels, where we assume $q$ is prime, are given by:
\begin{equation}
 X_q|j\rangle=|(j+1)\mod q\rangle,\quad Z_q|j\rangle=\omega^j|j\rangle
\end{equation}
with $\omega=e^{2\pi i/q}$, where $j\in\mathbb{Z}_q$. These Paulis form a group, denoted $\mathbb{P}_q$.
\end{definition}

When $q=2$, these are the standard qubit operators. This group structure is preserved over tensor products since each of these Paulis has order $q$.

\begin{definition}\label{condit}
An $n$-qudit stabilizer $s$ is an $n$-fold tensor of generalized Pauli operators, such that there exists at least one state, $|\psi\rangle$ such that:
\begin{equation}
s|\psi\rangle=|\psi\rangle
\end{equation}
where $|\psi\rangle\in \mathbb{C}^{q^n}$. 
\end{definition}

\begin{definition}
A stabilizer group $\bf{S}$ with commuting \textit{generators} $\{s_i \}$ is defined as the subgroup of all $n$-qudit generalized Paulis formed from all multiplicative compositions ($\circ$) of these generators. This subgroup must not contain a non-trivial multiple of the identity.
\end{definition}

\begin{definition}
We call a basis of orthonormal states $|\psi\rangle$ which satisfy the condition in Definition \ref{condit} for a stabilizer group $\bf{S}$ the codewords of the stabilizer.
\end{definition}

Since each operator has order $q$, a collection of $k$ compositionally independent generators for this stabilizer group will have $q^k$ elements.

Measuring the eigenvalues of the members in our stabilizer group, called the \textit{syndrome}, of our state gives us a way to determine what error might have occurred and then undo the determined error.

We recall for the reader, the well-known result:

\begin{theorem}For any stabilizer code with $k$ qudit stabilizers and $n$ physical qudits, there will be $q^{n-k}$ mutually orthogonal basis stabilizer states, or codewords.\label{words}
\end{theorem}

This differs from the standard convention of $k$ being the number of encoded qudits since throughout this work we focus ourselves on the number of stabilizer generators.

When discussing the errors that occur to our system, the standard choice of the depolarizing channel model focuses on the weights of the errors:

\begin{definition}
The weight of an $n$-qudit operator is given by the number of non-identity operators in it.
\end{definition}

\begin{definition}
A stabilizer code, specified by its stabilizers and stabilizer states, is characterized by a set of values:
\begin{itemize}
\item $n$: the number of qudits that the states are over
\item $n-k$: the number of encoded (logical) qudits, where $k$ is the number of stabilizers
\item $d$ (for non-degenerate codes (where all stabilizer group members have weight at least $d$)): the distance of the code, given by the lowest weight of an undetectable generalized Pauli error (commutes with all stabilizer generators)
\end{itemize}
These values are specified for a particular code as: $[[n, n-k,d]]_q$, where $q$ is the dimension of the qudit space.
\end{definition}

We note that, so long as no ambiguity exists, we suppress $\otimes$. We only include $\otimes$ to make register changes explicit.

Working with tensors of operators can be challenging, and so we make use of the following well-known mapping from these to vectors. This mapping is sometimes referred to as the symplectic representation, but we use alternative notation in this work to provide some notational flexibility utilized in this work. By representing these operators as vectors at times the solution to a problem can become far more tractable.

\begin{definition}[$\phi$ representation of a qudit operator]
We define the surjective map: 
\begin{equation}
\phi_q: \mathbb{P}_q^n\mapsto \mathbb{Z}_q^{2n}
\end{equation}
which carries an $n$-qudit Pauli in $\mathbb{P}_q^n$ to a $2n$ vector mod $q$, where we define this map as:
\begin{multline}
\phi_q(\omega^\alpha \otimes_{i-1} I\otimes X_q^a Z_q^b\otimes_{n-i} I)\\
=( 0^{ i-1}\ a\ 0^{ n-i} | 0^{ i-1}\ b\ 0^{ n-i}).
\end{multline}
Throughout we will assume that $\mathbb{Z}_q$ takes values in $\{0,\ldots,q-1\}$. This mapping is defined as a homomorphism with: $\phi_q(s_1\circ s_2)=\phi_q(s_1)\oplus \phi_q(s_2)$, where $\oplus$ is component-wise addition mod $q$. We denote the first half of the vector as $\phi_{q,x}$ and the second half as $\phi_{q,z}$.
\end{definition}

We may invert the map $\phi_q$ to return to the original $n$-qudit Pauli operator with the global phase being undetermined. We make note of a special case of the $\phi$ representation:

\begin{definition}
Let $q$ be the dimension of the initial system. Then we denote by $\phi_\infty$ the mapping:
\begin{equation}
    \phi_\infty:  \mathbb{P}_q^n\mapsto \mathbb{Z}^{2n}
\end{equation}
where no longer are any operations taken $\mod$ some base, but instead carried over the integers.
\end{definition}

The ability to define $\phi_\infty$ as a homomorphism still (and with the same rule) is a portion of the results of this paper--shown in Theorem \ref{inv}. In general we will write a stabilizer as $\phi_q$, perform some operations, then write it in $\phi_\infty$. We shorten this to write it as $\phi_\infty$, and can later select to write it as $\phi_{q'}$ for some prime $q'$ by taking element-wise $\mod q'$. When we provide no subscript for the representation, that implies that the choice is irrelevant.

The commutator of two operators in this picture is given by the following definition:
\begin{definition}
Let $s_i,s_j$ be two qudit Pauli operators over $q$ bases, then these commute if and only if:
\begin{equation}
\phi_q(s_i)\odot \phi_q(s_j)=0\mod q
\end{equation}
where $\odot$ is the symplectic product, defined by:
\begin{multline}
\phi_q(s_i)\odot \phi_q(s_j)\\ =\oplus_k [\phi_{q,z}(s_j)_k\cdot  \phi_{q,x}(s_i)_k- \phi_{q,x}(s_j)_k \cdot \phi_{q,z}(s_i)_k]
\end{multline}
where $\cdot$ is standard integer multiplication $\mod q$ and $\oplus$ is addition $\mod q$.
\end{definition}

Before finishing, we make a brief list of some possible operations we can perform on our $\phi$ representation for a stabilizer group:
\begin{enumerate}
    \item As remarked above, we may add rows of the stabilizer generator matrix together, which corresponds to composition of operators
    \item We may swap rows, corresponding to permuting the stabilizers
    \item We may multiply each row by any number in $\{1,\ldots,q-1\}$, corresponding to composing a stabilizer with itself. Since all operations are done over a prime number of bases, each number has an inverse.
    \item We may swap registers (qudits) in the following ways:
        \begin{enumerate}
            \item We may swap columns $(Reg\ i,Reg\ i+n)$ and $(Reg\ j,Reg\ j+n)$ for $0<i,j\leq n$, corresponding to relabelling qudits.
            \item We may swap columns $Reg\ i$ and $(-1)\cdot Reg\ i+n$, for $0<i\leq n$, corresponding to conjugating by a Hadamard gate on register $i$ (or Discrete Fourier Transforms in the qudit case \cite{qudit}) thus swapping $X$ and $Z$'s roles on that qudit.
        \end{enumerate}
\end{enumerate}

All of these operations leave all parameters of the code alone, but can be used in proofs. At this point we have all the necessary definitions to prove our results and have a solid base in qudit operators.

\section{Embedding Theorem}

In this section we begin by defining \textit{invariant} codes, which are codes that can be used for systems over any number of bases. Prior to this, only a few examples of invariant codes were known. Then we proceed to show that all qudit codes are invariant codes. This only shows that codes are valid over other spaces, so we then show that at least for sufficiently sized spaces all parameters of the code--particularly the distance--is at least preserved, if not even improved. We provide an argument about when the distance of the code will be improved. We finish by showing how to find the corresponding logical operators for these codes.



\begin{definition}[Invariant codes]
A stabilizer code is invariant iff:
\begin{equation}
    \phi_q(s_i)\odot \phi_q(s_j)=0,\quad \forall i,j
\end{equation}
holds for all primes $q$.
\end{definition}

This is satisfied if $\phi_\infty(s_i)\odot \phi_\infty(s_j)=0$, for all stabilizers $s_i$ and $s_j$ in the stabilizer group $\mathbf{S}$.

\subsection{Motivating Examples}
Consider the following example of generators for a stabilizer group: $\langle XX,ZZ\rangle$. As a qubit code this forms a valid stabilizer code with codeword:
\begin{equation}
    \frac{|00\rangle+|11\rangle}{\sqrt{2}}
\end{equation}
and the commutator of these generators can be seen to be: $(1)+(1)=2\equiv 0\mod 2$. Now suppose we wish to use this code for a qutrit system. In order to do that we must transform these generators into ones which have commutator 0, this can be achieved with $\langle XX^{-1},ZZ\rangle$. In this case $\phi_\infty (X\otimes X^{-1})\odot \phi_\infty(Z\otimes Z)=0$. This means that not only can this be used for qutrits, but for all prime number of bases. The codeword in the qutrit case is:
\begin{equation}
    \frac{|00\rangle+|12\rangle+|21\rangle}{\sqrt{3}}
\end{equation}
and the generalization of this for the codewords of a $q$ level system is a simple extension. We simply make each term in the codeword have the entries sum to a multiple of the qudit dimension so that the $ZZ$ operator has a $+1$ eigenvalue:
\begin{equation}
    \frac{1}{\sqrt{q}}(\sum_{j=1}^q |j\mod q, q-j\mod q\rangle).
\end{equation}
If we look at the generators of this code, there is no single qudit operator that commutes with the generators, thus the distance of this invariant form of the code is still $d=2$.

This is not the only example of a code that can be turned into invariant form. Another great example is the 5-qubit code \cite{5qudit}. In fact, no changes are needed:
\begin{equation}
    \langle XZZXI,\quad IXZZX,\quad XIXZZ,\quad ZXIXZ\rangle.
\end{equation}
From inspection this can be seen to have commutators 0, and so this is a valid stabilizer code for qudits, and it can also be checked that this code will always have distance 3.

It is helpful to have a couple of examples, however, it has been unknown whether it is always possible to put stabilizer codes into an invariant form. We move forward from here to show that this can always be done, and a method of how to do this. 

\subsection{Embedding Theorem Statement and Proof}

We now show that all qudit stabilizer codes can be written in an invariant form\footnote{We acknowledge Andrew Jena for his contributions in the form of the below theorem and corollary.}. This shows that we can form valid stabilizer groups over any number of bases, but says nothing about the distance of these codes. This aspect is treated in the section immediately following.\\

\begin{theorem}\label{inv}
All qudit stabilizer codes can be transformed into invariant codes.
\end{theorem}
\begin{proof}
Let $\{s_1,\dotso,s_k\}$ be a set of stabilizer generators for a qudit code over $q$ levels, with $k \leq n$ and $q$ prime. We must construct a set of stabilizers, $\{s_1',\dotso,s_k'\}$, such that:
\begin{enumerate}
    \item $\phi_\infty(s_i') \equiv \phi_q(s_i) \mod{q}$, for all $i$
    \item $\phi_\infty(s_i') \odot \phi_\infty(s_j') = 0$, for all $i \neq j$.
\end{enumerate}
Without loss of generality, we assume that our stabilizers are given in canonical form:
\begin{equation}
\begin{pmatrix}
\phi(s_1) \\
\vdots \\
\phi(s_k)
\end{pmatrix} = 
\left(
\begin{array}{c|c}
I_{k}\ X_2 & Z_1\ Z_2
\end{array}
\right).
\end{equation}
We define the strictly lower diagonal matrix, $L$, with entries:
\begin{equation}
L_{ij} = \begin{cases}
0 & i \leq j \\
\phi(s_i) \odot \phi(s_j) & i > j
\end{cases}
\end{equation}
and define $s_1',\dotso,s_k'$ such that:
\begin{equation}
\begin{pmatrix}
\phi(s_1') \\
\vdots \\
\phi(s_k')
\end{pmatrix} = 
\left(
\begin{array}{c|c}
I_{k}\ X_2 & Z_1+L\quad Z_2
\end{array}
\right).
\end{equation}
We show that $s_1',\dotso,s_k'$ satisfy the conditions.
\begin{enumerate}
    \item Since $\phi(s_i) \odot \phi(s_j) \equiv 0\mod q$ for all $i \neq j$, we observe that $L_{ij} \equiv 0 \mod q$ for all entries. By adding rows of $L$ to our stabilizers, we have not changed the code modulo $q$.
    \item For $i > j$, we observe that:
    \begin{align*}
        &\phi(s_i') \odot \phi(s_j') \\
        &= (\phi(s_i) + (0\ |\ L_i\ 0)) \odot (\phi(s_j) + (0\ |\ L_j\ 0)) \\
        &= \phi(s_i) \odot \phi(s_j) + \phi(s_i) \odot (0\ |\ L_j\ 0) \\
        &\hspace{1em}+ (0\ |\ L_i\ 0) \odot \phi(s_j) + (0\ |\ L_i\ 0) \odot (0\ |\ L_j\ 0) \\
        &= \phi(s_i) \odot \phi(s_j) + 0 - L_{ij} + 0 \\
        &= 0.
    \end{align*}
\end{enumerate}
\end{proof}
\begin{widetext}
\begin{example}
Consider the 7-qubit Steane code with parameters $[[7,1,3]]_2$, denote it by $\Xi$ \cite{steane}. The $\phi$ representation is given by:
\begin{equation}
\phi_2(\Xi)=\begin{bmatrix}
H & | & 0\\
0 & | & H
\end{bmatrix}
\end{equation}
where $H$ is the parity-check matrix for the classical Hamming code given by:
\begin{equation}
H=\begin{bmatrix}
 1 & 0 & 0 & 1 & 0 & 1 & 1\\
 0 & 1 & 0 & 1 & 1 & 0 & 1\\
 0 & 0 & 1 & 0 & 1 & 1 & 1
\end{bmatrix}.
\end{equation}
We begin by putting this in standard form, performing operations $\mod 2$: 
\begin{equation}
\phi_2(\Xi)=
\setcounter{MaxMatrixCols}{15}
\begin{bmatrix}
 1 & 0 & 0 & 0 & 0 & 0 & 1 & | & 0 & 0 & 0 & 0 & 1 & 1 & 0\\
 0 & 1 & 0 & 0 & 0 & 0 & 0 & | & 0 & 0 & 0 & 1 & 1 & 1 & 0\\
 0 & 0 & 1 & 0 & 0 & 0 & 1 & | & 0 & 0 & 0 & 1 & 0 & 1 & 0\\
 0 & 0 & 0 & 1 & 0 & 0 & 0 & | & 1 & 1 & 0 & 0 & 0 & 0 & 1\\
 0 & 0 & 0 & 0 & 1 & 0 & 0 & | & 0 & 1 & 1 & 0 & 0 & 0 & 1\\
 0 & 0 & 0 & 0 & 0 & 1 & 0 & | & 1 & 1 & 1 & 0 & 0 & 0 & 0\\
\end{bmatrix}.
\end{equation}
 For the following operations, we no longer take our operations over $\mod 2$.

The anti-symmetric matrix $[\odot]$ representing the symplectic inner products between the stabilizers and the resulting $L$ matrix for this code are given below:
\begin{equation}
[\odot]=
\begin{bmatrix}
 0 & 0 & 0 & 2 & 0 & 0\\
 0 & 0 & 0 & 0 & 0 & 0\\
 0 & 0 & 0 & 0 & 2 & 0\\
 -2 & 0 & 0 & 0 & 0 & 0\\
 0 & 0 & -2 & 0 & 0 & 0\\
 0 & 0 & 0 & 0 & 0 & 0
\end{bmatrix}\Rightarrow 
L=
\begin{bmatrix}
 0 & 0 & 0 & 0 & 0 & 0\\
 0 & 0 & 0 & 0 & 0 & 0\\
 0 & 0 & 0 & 0 & 0 & 0\\
 -2 & 0 & 0 & 0 & 0 & 0\\
 0 & 0 & -2 & 0 & 0 & 0\\
 0 & 0 & 0 & 0 & 0 & 0
\end{bmatrix}.
\end{equation}
Adding this to our standard form, we have an invariant form for the Steane code given by:
\begin{equation}
\phi_\infty(\Xi)=
\setcounter{MaxMatrixCols}{15}
\begin{bmatrix}
 1 & 0 & 0 & 0 & 0 & 0 & 1 & | & 0 & 0 & 0 & 0 & 1 & 1 & 0\\
 0 & 1 & 0 & 0 & 0 & 0 & 0 & | & 0 & 0 & 0 & 1 & 1 & 1 & 0\\
 0 & 0 & 1 & 0 & 0 & 0 & 1 & | & 0 & 0 & 0 & 1 & 0 & 1 & 0\\
 0 & 0 & 0 & 1 & 0 & 0 & 0 & | & -1 & 1 & 0 & 0 & 0 & 0 & 1\\
 0 & 0 & 0 & 0 & 1 & 0 & 0 & | & 0 & 1 & -1 & 0 & 0 & 0 & 1\\
 0 & 0 & 0 & 0 & 0 & 1 & 0 & | & 1 & 1 & 1 & 0 & 0 & 0 & 0\\
\end{bmatrix}.
\end{equation}
\end{example}
\end{widetext}
We will want to know the size of the maximal entry in this invariant form for our bound on ensuring the distance of the code is at least preserved. The bound on the maximal entry is provided from the above proof:

\begin{corollary}\label{bound}
The maximal element in $\phi_\infty(S)$, $B$, is upper bounded by:
\begin{equation}
    (2+(n-k)(q-1))(q-1).
\end{equation}
\end{corollary}
\begin{proof}
For any $i \neq j$, there are at most $n-k$ entries in which both $\phi_{q,x}(s_i)$ and $\phi_{q,z}(s_j)$ are non-zero and bounded above by $q-1$, and a single entry in which one is 1 whereas the other is bounded above by $q-1$. This gives us a bound on the inner product of: $(n-k)(q-1)^2 + (q-1)$. This is a bound on the size of an entry in our invariant stabilizer of $q-1 + (n-k)(q-1)^2 + (q-1) = (2+(n-k)(q-1))(q-1)$.
\end{proof}

\begin{widetext}
\begin{example}
In this example we show that CSS codes remain CSS codes under this transformation. Consider a general CSS code given by:
\begin{equation}
\phi(\Xi)=\begin{bmatrix}
I_{k_1} & X_{k_2} & X_{n-(k_1+k_2)} & | & 0 & 0 & 0\\
0 & 0 & 0 & | & Z_{k_1} & I_{k_2} & Z_{n-(k_1+k_2)}
\end{bmatrix}
\end{equation}
where we have put the two block matrices into approximately standard form. Now, we perform Hadamards (or discrete Fourier transforms) on the $k_2$ sized middle blocks. We then have:
\begin{equation}
\phi(\Xi)=\begin{bmatrix}
I_{k_1} & 0 & X_{n-(k_1+k_2)} & | & 0 & X_{k_2} & 0\\
0 & I_{k_2} & 0 & | & Z_{k_1} & 0 & Z_{n-(k_1+k_2)}
\end{bmatrix}.
\end{equation}
Now, we note that the first $k_1$ stabilizers exactly commute with each other, i.e., inner product 0 in the $\phi_\infty$ sense, and likewise for the $k_2$ other stabilizer generators. Now we simply need to consider the case where we pick generators from each of the halves. We consider the matrix $[\odot]$, as above. This has nonzero entries for rows in $k_2$ when the columns are in $k_1$. Likewise for when the rows are in $k_1$, the entries are nonzero for columns in $k_2$. Thus we only add entries to $Z_{k_1}$ and $X_{k_2}$ with $[\odot]$ and, hence, certainly also for our $L$ matrix. In fact, the L matrix adds entries only to $Z_{k_1}$ since it is lower triangular. Given the new invariant form matrix, we may now invert our initial step of applying discrete Fourier transforms and we will still have a CSS code.
\end{example}
\end{widetext}

\subsection{Distance Preserving Condition}

Now that we know that all qudit codes can be put into an invariant form, we now prove that at least for most sizes of the space we can ensure that the distance of the code is at least preserved. We find a cutoff on the number of bases in the underlying space needed to at least preserve the distance.

\begin{theorem}
For all primes $p>p^*$, with $p^*$ a cutoff value greater than $q$, the distance of an embedding of a non-degenerate stabilizer code $[[n,n-k,d]]_q$ into $p$ bases, $[[n,n-k,d']]_p$, has $d'\geq d$.
\end{theorem}

Before proving this theorem we make a couple of nuanced definitions:
\begin{definition}
An unavoidable error is an error that commutes with all stabilizers and produces the $\vec{0}$ syndrome over the integers.
\end{definition}

These correspond to undetectable errors that would remain undetectable regardless of the number of bases for the code since they always exactly commute under the symplectic inner product with all stabilizer generators--and so all members of the stabilizer group. Since these errors are always undetectable we call them unavoidable errors since changing the number of bases would not allow this code to detect this error. This then provides the following insight:

\begin{remark}
The distance of a code over the integers is given by the minimal weight member in the set of unavoidable errors. The distance over the integers is represented by $d^*$, and so $d^*\geq d$. This value is also the minimum number of columns of the stabilizer generator matrix that are linearly dependent over the integers (or equivalently over the rationals), in the symplectic sense.
\end{remark}

We also define the other possible kind of undetectable error for a given number of bases, which corresponds to the case where some syndromes are multiples of the number of bases:

\begin{definition}
An artifact error is an error that commutes with all stabilizers but produces at least one syndrome that is only zero modulo the base.
\end{definition}

These are named artifact errors as their undetectability is an artifact of the number of bases selected and could become detectable if a different number of bases were used with this code. Each undetectable error is either an unavoidable error or an artifact error. We utilize this fact to show our theorem.

\begin{proof}
The ordering of the stabilizers and the ordering of the registers does not alter the distance of the code. With this, $\phi_\infty$ for the stabilizer generators over the integers can have the rows and columns arbitrarily swapped.

Let us begin with a code over $q$ bases and extend it to $p$ bases. The errors for the original code are the vectors in the kernel of $\phi_q$ for the code. These errors are either unavoidable errors or are artifact errors. We may rearrange the rows and columns so that the stabilizers and registers that generate these entries that are nonzero multiples of $q$ are the upper left $2d\times 2d$ minor, padding with identities if needed. The factor of 2 occurs due to the number of nonzero entries in $\phi_\infty$ being up to double the weight of the Pauli. The stabilizer(s) that generate these multiples of $q$ entries in the syndrome are members of the null space of the minor formed using the corresponding stabilizer(s).

Now, consider the extension of the code to $p$ bases. Building up the qudit Pauli operators by weight $j$, we consider the minors of the matrix composed through all row and column swaps. These minors of size $2j\times 2j$ can have a nontrivial null space in two possible ways:
\begin{itemize}
    \item If the determinant is 0 over the integers then this is either an unavoidable error or an error whose existence did not occur due to the choice of the number of bases.
    \item If the determinant is not 0 over the integers, but takes the value of some multiple of $p$, then it's $0\mod p$ and so a null space exists.
\end{itemize}
Thus we can only introduce artifact errors to decrease the distance. By bounding the determinant by $p^*$, any choice of $p>p^*$ will ensure that the determinant is a unit in $\mathbb{Z}_p$, and hence have a trivial null space since the matrix is invertible.

Now, in order to guarantee that the value of $p$ is at least as large as the determinant, we can use Hadamard's inequality to obtain:
\begin{equation}
    p> p^* =B^{2(d-1)}(2(d-1))^{(d-1)}
\end{equation}
where $B$ is the maximal entry in $\phi_\infty$. Since we only need to ensure that the artifact induced null space is trivial for Paulis with weight less than $d$, we used this identity with $2(d-1)\times 2(d-1)$ matrices.

When $j=d$, we can either encounter an unavoidable error, in which case the distance of the code is $d$ or we could obtain an artifact error, also causing the distance to be $d$. It is possible that neither of these occur at $j=d$, in which case the distance becomes some $d'$ with $d<d'\leq d^*$. 
\end{proof}

\begin{example}
In our example of the Steane code, we have $B=1$ and $d=3$, so for all primes larger than $1^{2\cdot 2}(2\cdot 2)^2=16$ we are guaranteed that the distance is preserved. For primes below that value, we can manually check and apply alternate manipulations if needed. Given that all entries are $\pm 1$, we know that the determinant of all the minors is bounded by $4$, all primes at least as large as $5$ preserve the distance. Through manual checking 3 also is not a possible minor determinant, so all primes preserve the distance for our invariant form of the Steane code.
\end{example}

We alluded prior to this proof that the code over the integers has distance at least as large. To determine how many bases are needed to ensure we have distance $d^*$, we simply extend our above result to obtain the cutoff expression, whereby no further distance improvements can be obtained from embedding the code--suggesting that another code ought to be used.

\begin{corollary}
For a non-degenerate stabilizer code we obtain the integer distance $d^*$ when:
\begin{equation}
    p> B^{2(d^*-1)} (2(d^*-1))^{d^*-1}.
\end{equation}
After this value the distance cannot be improved through embedding. If $d^*$ is unknown, this can be upper bounded by using $k$ in place of $d^*$.
\end{corollary}

\begin{proof}
This follows from the above proof. The looser bound comes from $d^*\leq k$, so we can evaluate this at $d^*=k$ to obtain the loosest condition.
\end{proof}


The above provides a condition on the number of bases needed to ensure the distance of the code is at least preserved, but one could also ask, given an invariant code, whether that code can be used over fewer bases. We provide a bound on this with the following:
\begin{lemma}
For a non-degenerate code, for all $p<p^{**}$, with $p^{**}$ a cutoff value less than $q$ (possibly $\leq 2$), the distance of $[[n,n-k,d]]_q$ over $p$ bases, $[[n,n-k,d']]_p$, must have $d'<d$.
\end{lemma}

\begin{proof}
Let $t=\lfloor \frac{d-1}{2}\rfloor$. The qudit quantum Hamming bound requires the initial code to satisfy:
\begin{equation}
    \sum_{j=0}^t {n\choose j}(q^2-1)^j\leq q^k.
\end{equation}
Now we consider applying the code over $p$ levels. Then we may bound:
\begin{equation}
    {n\choose t}(p^2-1)^t\leq \sum_{j=0}^t {n\choose j}(p^2-1)^j.
\end{equation}
Likewise, when $p\geq 2$ we may bound:
\begin{equation}
    p^k\leq (p^2-1)^k.
\end{equation}
Combining these we have:
\begin{equation}
    {n\choose t}(p^2-1)^t\leq (p^2-1)^k.
\end{equation}
Then we violate the initial inequality if:
\begin{equation}
    p< \sqrt{1+{n\choose t}^{1/(k-t)}}=p^{**}
\end{equation}
This means that $p^{**}$ is only a valid bound when it is larger than 2, otherwise this result is trivially true since we no longer have a quantum code.
\end{proof}

Combining these results mean that distance \textit{may} be preserved for $p^{**}\leq p< p^*$, while for $p>p^*$ it is guaranteed to have the distance preserved. For the region of values of $p$ where the distance might be preserved, one can manually check and attempt another invariant form to try to make the distance preserved for the desired number of bases.

\subsection{Invariant Logical Operators}

Besides from the stabilizers we also need logical operators to perform computations over the encoded qudits. Now we show how to construct such invariant logical operators.

\begin{lemma}
We may define invariant logical operators, $\mathcal{L}_\infty$, for the stabilizer code $\mathbf{S}$ as well.
\end{lemma}

\begin{proof}
Each logical operator is in $N(\mathbf{S})/\mathbf{S}$, the normalizer of $\mathbf{S}$ excluding $\mathbf{S}$, and there are $n-k$ $X$ logical operators and $n-k$ $Z$ logical operators. This means that we could, if we desired, generate a code $\mathbf{S}'$ whose generators are $\mathbf{S}\cup L_X$. This will have rank $n$ and can be written in standard form as:
\begin{equation}
    \begin{bmatrix}
      I_n | *
    \end{bmatrix}
\end{equation}
meaning that $L_X$ may be diagonalized within the last $n-k$ qudits. This can also be done with $L_Z$.

Then, since these logical operators are compositionally independent, they must be linearly independent in the $\phi$ representation, meaning $rank(L_X\cup L_Z)=2(n-k)$. Now, if we take the standard form for $\mathbf{S}$ and append $L_X,L_Z$ as additional rows we have:
\begin{equation}
\begin{bmatrix}
  \mathbf{S}\\
  L_X\\
  L_Z
\end{bmatrix}=
    \begin{bmatrix}
       I_k & A & | & B & C\\
       0 & D & | & E & F\\
       0 & G & | & H & J
    \end{bmatrix}
\end{equation}
From the above observation it is possible to compose the generators for $L_X,L_Z$ to generate the matrix:
\begin{equation}
    \begin{bmatrix}
     I_k & A & | & B & C\\
     0 & I_{n-k} & | & E' & F'\\
     0 & G' & | & H' & I_{n-k}
    \end{bmatrix}
\end{equation}
At this point we focus on fixing the commutators between elements of $L_X$ and $L_Z$. Since the first $k$ qudits will always contribute 0 to the commutator we drop those columns:
\begin{equation}
    \begin{bmatrix}
      I_{n-k} & |  & F'\\
      G' & |  & I_{n-k}
    \end{bmatrix}
\end{equation}
We can further reduce this to:
\begin{equation}
    \begin{bmatrix}
      I_{n-k} & | & 0\\
      0 & | & I_{n-k}
    \end{bmatrix}
\end{equation}
This trivially satisfies the required relations:
\begin{equation}
    \phi_q(\bar{X}_i)\odot \phi_q(\bar{Z}_j)=\delta_{ij}
\end{equation}
\begin{equation}
    \phi_q(\bar{X}_i)\odot \phi_q(\bar{X}_j)=\phi_q(\bar{Z}_i)\odot \phi_q(\bar{Z}_j)=0,\quad \forall i,j
\end{equation}
Throughout these computations we have updated $E'$ and $H'$. We now simply apply Theorem \ref{inv} to each logical operator in turn appended to $\phi(\mathbf{S})$.
\end{proof}

\begin{remark}
This process does not alter our invariant stabilizer form, so our bound from earlier still holds.
\end{remark}

\section{Conclusion and Discussion}

This work introduces and lays the groundwork for qudit codes that can be used on systems with local dimension different than initially designed. This helps ease the restrictions that some qudit codes suffer from. We showed one method for generating these invariant codes, but bring up the following example to motivate additional work on this:

\begin{example}

Throughout we have considered the method of creating invariant codes given by Theorem \ref{inv}. With the following simple example we can show that $p^*=q$ for this method is not always possible. Consider the $[[4,2,2]]_2$ code generated by:
\begin{equation}
    \Xi=\langle XZXX,ZXZZ\rangle
\end{equation}
Following the method prescribed we obtain:
\begin{equation}
\phi_\infty(\Xi)=\begin{bmatrix}
1 & 0 & 1 & 1  & | & 0 & 1 & 0 & 0 \\
0 & 1 & 0 & 0  & | & -3 & 0 & 1 & 1 \\
\end{bmatrix}
\end{equation}
This means that if we were to use this as a qutrit code the distance will drop to 1.
This cannot be resolved by changing the choices of generators through compositions. If, however, we select as our generators:
\begin{equation}
    \Xi'=\langle XZXX,ZXZZ^{-1}\rangle
\end{equation}
then $\Xi'$ is an invariant code and the distance of this code remains 2. Determining whether such a modification is always possible, and whether it's possible to achieve this with a simple procedure, are other open problems.
\end{example}

In this paper we have shown that qudit codes can be embedded into larger spaces, and at least for sufficiently large number of bases, all parameters of the code are at least preserved. This result provides another tool for error-correction schemes for qudit quantum computers by providing immediate codes for these devices using modifications of already known codes.

Although in this work we find some critical value, $p^*$, above which all primes preserve the distance of the code, we believe that this result carries to all primes at least as large as the initial dimension if one uses other procedures to make the code invariant. Proving this, or at least tightening the bound on the critical value, seems like an important extension of this result, since the current bound can be quite large. In addition, there is the question of whether these results also hold for degenerate codes.

Some additional directions to carry these results include the following. Determining whether the prescribed method for generating invariant codes, or some other method, allow for transversality preservation--a crucial tool in fault-tolerant quantum computation. We also ask whether it is possible to take codes already known over $q$ levels, and not a perfect code, and preserve the distance while using the code over $p<q$ levels. 

\section*{Acknowledgments}

We would like to thank David Cory for useful suggestions. We also thank Daniel Gottesman for reading over an earlier draft and providing some useful directions and caveats to consider. We also thank Andrew Jena for his help in proving the aforementioned theorem and corollary.

\section*{Funding}
We gratefully thank the financial contributions of the Canada First Research Excellence Fund, Industry Canada, CERC (215284), NSERC (RGPIN-418579), CIFAR, and the Province of Ontario.

\bibliographystyle{unsrt}
\phantomsection  
\renewcommand*{\bibname}{References}

\bibliography{main}

\end{document}